\documentclass[letterpaper,11pt]{article}

\usepackage[arXiv]{optional}

\usepackage{verbatim}
\usepackage[margin=1in]{geometry}
\usepackage{setspace}
\usepackage{graphicx}
\usepackage{epsfig}

\usepackage{microtype} 
\usepackage{array} 

\usepackage{amssymb}
\usepackage{amsmath}
\usepackage{amsthm}
\usepackage{mathtools} 
\usepackage{stmaryrd} 
\usepackage{stackrel} 
\usepackage{latexsym} 

\usepackage{tikz}
\usetikzlibrary{arrows}
\usetikzlibrary{arrows.meta}
\usetikzlibrary{decorations.pathreplacing}
\usetikzlibrary{cd}
\usepackage[all]{xy} 

\usepackage[ruled]{algorithm2e} 
\usepackage[noend]{algpseudocode}
\algnewcommand{\LineComment}[1]{\Statex\hspace{\algorithmicindent}\(\triangleright\) #1}
\algnewcommand\algorithmicforeach{\textbf{for each}}
\algdef{S}[FOR]{ForEach}[1]{\algorithmicforeach\ #1\ \algorithmicdo}
\usepackage{etoolbox} 

\usepackage{caption}
\usepackage{subcaption}

\usepackage{makecell} 
\usepackage{booktabs} 
\usepackage{multirow} 
\usepackage{tablefootnote} 

\usepackage[hidelinks]{hyperref} 
\usepackage[symbol]{footmisc}
\usepackage{makeidx} 
\usepackage{url} 
\usepackage{color}
\usepackage{xcolor}
\usepackage{xspace} 
\usepackage[normalem]{ulem} 
\usepackage{soul} 
\usepackage{ifpdf} 
\usepackage{extarrows}
\usepackage{stackengine}    
\usepackage{lipsum}
\usepackage{bm}
\usepackage{cleveref} 
\usepackage{cite} 
\usepackage{booktabs}
\makeatletter
\expandafter\patchcmd\csname\string\algorithmic\endcsname{\itemsep\z@}{\itemsep=0.25ex}{}{}
\makeatother

\theoremstyle{plain}
    \newtheorem{theorem}{Theorem}
    \newtheorem{proposition}[theorem]{Proposition}

    \newtheorem{corollary}[theorem]{Corollary}
    \newtheorem{fact}[theorem]{Fact}

    \newtheorem*{algr*}{Algorithm}
\theoremstyle{definition}
    \newtheorem{definition}[theorem]{Definition}

    \newtheorem{remark}[theorem]{Remark}
    \newtheorem*{remark*}{Remark}
    \newtheorem*{example*}{Example}
    
\usepackage{enumitem}
\setlist[itemize]{leftmargin=\parindent}
\setlist[enumerate]{leftmargin=\parindent}
\setlist[description]{font=\bfseries,leftmargin=\parindent}

\title{Linked Barcode for Persistence Induced by Filtrations\thanks{This research is partially supported by NSF grants CCF 2049010, CCF 2301360, and CCF 2439255.}}

\author{Tamal K. Dey\thanks{Department of Computer Science, Purdue University. \texttt{tamaldey@purdue.edu}}
\and Gilberto Gonzalez-Arroyo\thanks{Department of Computer Science, Purdue University. \texttt{gonza982@purdue.edu}}
\and Tao Hou\thanks{Department of Computer Science, University of Oregon. \texttt{taohou@uoregon.edu}}
}
\date{}

\crefname{equation}{Eq.}{Eqs.}

\newcommand{\tao}[1]{\textcolor{magenta}{TH: #1}}
\newcommand{\gilberto}[1]{\textcolor{red}{GG: #1}}

\newcommand{\para}[1]{\opt{arXiv}{\paragraph{#1}}\opt{SoCG}{\subparagraph*{#1}}}

\newcommand{\Hm}{\mathsf{H}}

\newcommand{\C}{\mathsf{C}}
\newcommand{\Z}{\mathsf{Z}}
\newcommand{\M}{\mathsf{M}}
\newcommand{\N}{\mathsf{N}}
\newcommand{\U}{\mathsf{U}}
\newcommand{\W}{\mathsf{W}}
\newcommand{\Y}{\mathsf{Y}}
\newcommand{\Ps}{\mathsf{P}}
\newcommand{\B}{\mathsf{B}}

\newcommand{\X}{\mathsf{X}}
\newcommand{\Lsf}{\mathsf{L}}

\newcommand{\Real}{\mathbb{R}}

\newcommand{\tran}{\mathsf{T}}
\renewcommand{\ker}{\mathsf{ker}}
\newcommand{\coker}{\mathsf{coker}}
\newcommand{\img}{\mathsf{img}}

\newcommand{\Pers}{\mathcal{PB}}
\newcommand{\lbarc}{\mathcal{LB}}
\newcommand{\lkbarc}{\mathcal{LKB}}

\renewcommand{\bar}[1]{\overline{#1}}

\newcommand{\lbarrowspace}{\;}

\let\rightarrowsp\rarrowsp

\newcommand{\incto}{\hookrightarrow}
\newcommand{\inctosp}[1]{\xhookrightarrow{\lbarrowspace#1\lbarrowspace}}

\newcommand{\sset}{\Sigma}
\newcommand{\pset}{\Pi}
\newcommand{\nset}{\Xi}
\newcommand{\gr}{\mathrm{gr}}
\newcommand{\subc}{\mathbf{subc}}
\newcommand{\lfunc}{\Phi}
\newcommand{\vecc}{\mathbf{vec}_\field}

\let\emptyset\varnothing

\newcommand{\Dcal}{\mathcal{D}}
\newcommand{\Fcal}{\mathcal{F}}

\newcommand{\Rcal}{\mathcal{R}}

\newcommand{\eG}{\epsilon}

\newcommand{\iG}{\iota}

\newcommand{\lG}{\lambda}

\newcommand{\sG}{\sigma}

\newcommand{\blue}{\color{blue}}
\newcommand{\gray}{\color{gray}}

\newcommand{\lbirth}{\mathsf{a}}
\newcommand{\ldeath}{\mathsf{b}}

\newcommand{\filt}{\Fcal}
\newcommand{\cplx}{K}

\newcommand{\filtr}{\Rcal}
\newcommand{\cplxr}{K^\mathcal{R}}
\newcommand{\ssxr}{\tau}
\newcommand{\subcplx}{L}

\newcommand{\field}{\Bbbk}
\newcommand{\low}{\mathrm {low}}
\newcommand{\col}{\mathrm {col}}
\newcommand{\row}{\mathrm {row}}

\makeatletter
\newcommand*{\da@rightarrow}{\mathchar"0\hexnumber@\symAMSa 4B }
\newcommand*{\da@leftarrow}{\mathchar"0\hexnumber@\symAMSa 4C }
\newcommand*{\xdashrightarrow}[2][]{%
  \mathrel{%
    \mathpalette{\da@xarrow{#1}{#2}{}\da@rightarrow{\;}{}}{}%
  }%
}
\newcommand{\xdashleftarrow}[2][]{%
  \mathrel{%
    \mathpalette{\da@xarrow{#1}{#2}\da@leftarrow{}{}{\;}}{}%
  }%
}
\newcommand{\xdashleftrightarrow}[2][]{%
  \mathrel{%
    \mathpalette{\da@xarrow{#1}{#2}\da@leftarrow\da@rightarrow{}{}}{}%
  }%
}
\newcommand*{\da@xarrow}[7]{%
  \sbox0{$\ifx#7\scriptstyle\scriptscriptstyle\else\scriptstyle\fi#5#1#6\m@th$}%
  \sbox2{$\ifx#7\scriptstyle\scriptscriptstyle\else\scriptstyle\fi#5#2#6\m@th$}%
  \sbox4{$#7\dabar@\m@th$}%
  \dimen@=\wd0 %
  \ifdim\wd2 >\dimen@
    \dimen@=\wd2 %
  \fi
  \count@=2 %
  \def\da@bars{\dabar@\dabar@}%
  \@whiledim\count@\wd4<\dimen@\do{%
    \advance\count@\@ne
    \expandafter\def\expandafter\da@bars\expandafter{%
      \da@bars
      \dabar@ 
    }%
  }%
  \mathrel{#3}%
  \mathrel{%
    \mathop{\da@bars}\limits
    \ifx\\#1\\%
    \else
      _{\copy0}%
    \fi
    \ifx\\#2\\%
    \else
      ^{\copy2}%
    \fi
  }%
  \mathrel{#4}%
  \!\!
}
\makeatother

\newcounter{desccounter}

\newcommand{\cancel}[1]

\begin{document}

\date{}

\maketitle

\begin{abstract}
The well-known persistence algorithm
summarizes the evolution of homological cycles into what is called a \emph{barcode}
while scanning an input simplicial
filtration. We show that this summarization
process can be enriched by monitoring other algebraic structures that
weave through different dimensions. In particular, we propose an algorithm
to monitor the $(p+1)$-chains that make $p$-cycles to be $p$-boundaries and then
morph into $(p+1)$-cycles. In effect, we get extra bars called \emph{links} connecting the
bars in dimension $p$ with the bars in dimension $p+1$ in the persistence barcode. 
The links produce extra
barcodes, which we call \emph{link barcodes} 
in addition to the usual ones obtained by standard persistence. 
The link barcodes, as such, are not stable. However, we can make them stable
using a fixed ``reference'' filtration. We apply
the link barcodes
to the graph isomorphism 
problem and to the link prediction problem in temporal networks exhibiting its discriminating power through these experiments.
\end{abstract}



\section{Introduction}
The persistence algorithm as originally proposed~\cite{ELZ02} works on
an input simplicial filtration that is a growing sequence of simplicial
complexes. This ``geometric'' input inducing an ``algebraic'' input called
\emph{persistence module}~\cite{ZC05} is prevalent in practice ~\cite{chazal2021introduction}. There is
a simple intuitive interpretation of persistence in terms of the evolution
of homological cycles (their classes) through the growing spaces in the
input filtration. Persistence tracks \emph{independent} classes of
such cycles recording their survival times delimited by their creation
times called \emph{births} and destruction times called \emph{deaths}.
This record is output as a set of intervals called bars that comprise the \emph{barcode} summarizing the input filtration. 

In this paper, we enrich the above summarization of the input with a barcode that captures evolution of other topological features in addition to the homological cycles. Interestingly, these features interconnect with cycles to give a barcode
that weaves through different dimensions by connecting the bars in the
classical barcode. A homological $p$-cycle $z$ dies when a $(p+1)$-chain $y$ gets
completed where $z$ bounds $y$. The classical persistence algorithm detects this
and does not consider $y$ any further. We continue to monitor $y$, which gets born at
the death of $z$ and survives until it evolves to a $(p+1)$-cycle, say $z'$, with additions of more $(p+1)$-simplices. This provides an interval (\emph{link}) for the
$(p+1)$-chain $y$ from the death of the $p$-cycle $z$ to the birth of the
$(p+1)$-cycle $z'$. In effect, this link connects a bar in $p$-dimension to a bar
in $(p+1)$-dimension in the classical barcode. Consequently, we now get a bar 
that is obtained by connecting the bars for homological cycles in different
dimensions in a canonical way giving us a new barcode we call \emph{linked barcode}; see
Figure~\ref{fig:linkedbar}.

First, we observe that the links themselves arise from the persistence module induced by a simplex-wise
filtration in a unique way. 
This confirms that links provide a canonical
connection among the homological bars. Next, we design an algorithm for computing
these links, the  justification of which is derived from an algebraic
construct called \emph{presentations} of persistence modules~\cite{lesnicknotes,NVO00}.
The algorithm works on the classic matrix reductions with left-to-right
column additions as employed in the standard persistence algorithm.
The matrix-based persistence algorithm~\cite{EH10} only computes
the reduced matrix $R$ from a filtered boundary matrix $D$ where it implicitly
also computes an upper triangular matrix $V$ with $R=DV$. For linked barcode,
we take the matrix $V$, crop it as algebra dictates, and reduce it further.
Surprisingly, this simple adaptation/addition to the classical left-to-right
column additions provides a more enriched summarization of the input
which has so far been missed in TDA literature. The algorithm runs in $O(n^3)$ time for a filtration that inserts $n$ simplices altogether.

We adapt our basic algorithm to the special
case of graph filtrations which is abundant in applications~\cite{carriere2020perslay,DBLP:conf/iclr/HornBMMRB22,zhang2022gefl,pmlr-v108-zhao20d}. For this case,
taking advantage of the special structures of graphs, we get an
$O(n\log n)$ time algorithm using some dynamic data structures as in~\cite{DHP23}. 
The link barcode as such is not stable at the module level. 
However, we show that it can be made
stable by considering a ``reference'' filtration. This means that the links 
as connectors between persistence bars with reference to another filtration
remain stable.

Finally, we present two experimental results to illustrate that linked barcode
indeed has more discriminating power in practice than the standard homological barcode. Firstly, we define distances on the set of graphs using both the link barcode and the standard persistence barcode. We then show the existence of non-isomorphic graphs that are indistinguishable by their persistence barcodes, yet can be distinguished by their link barcodes. Secondly, we perform temporal link prediction by extracting local structural features around vertex pairs via both barcodes. Our results demonstrate that utilizing the link barcode yields better predictive accuracy than relying solely on the persistence barcode.


\begin{figure}[htbp]
\centerline{\includegraphics[width=\textwidth]{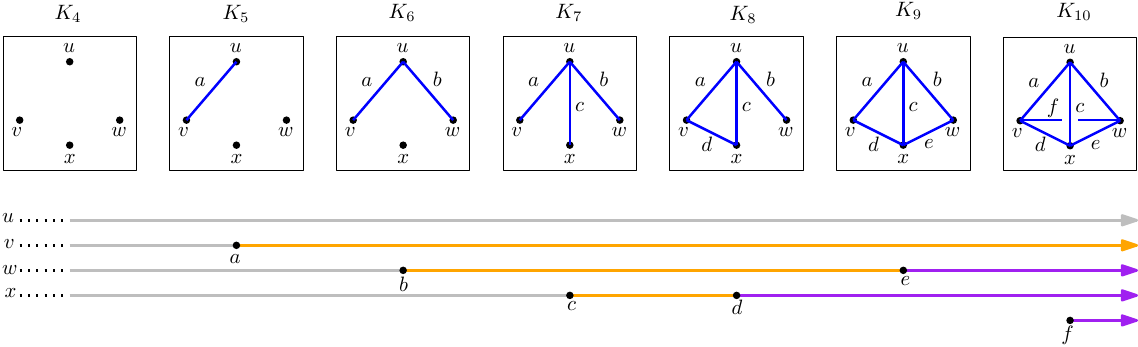}}
\caption{$1$-link barcode of $\Lsf_1$ (orange) connecting barcodes of $\Hm_0$ (gray) and $\Hm_1$ (purple); see section~\ref{sec:linkedbarcode} for
definitions of link module $\Lsf_p$ and $p$-link barcode.}
\label{fig:linkedbar}
\end{figure}

\section{Persistence modules}
We briefly review
definitions of persistence modules  as quiver representations;
see, e.g.~\cite{Oudot15,Schiffler14}
for more details.
A \emph{persistence module} (or \emph{module} for short), as depicted in the top row of~\eqref{MPM}, is a functor $\M\colon \mathbb{N}_0\rightarrow \mathbf{vec_{\field}}$, where $\mathbf{vec_{\field}}$ denotes the category of \emph{finite}-dimensional vector spaces over a field $\field$. It is of \emph{finite type}, if there exists an $m\in\mathbb{N}_0$ such that $\M(i\leq j)$ is an isomorphism for all $i\geq m$. 
In this paper, we only consider persistence modules of finite type.
We sometimes truncate such a module by showing only the first 
$m$ vector spaces and their connecting morphisms;
see, e.g.,~\eqref{eqn:C}--\eqref{eqn:B}.

A morphism of persistence modules $\phi\colon \M\rightarrow \N$, as depicted in~\eqref{MPM},
\begin{equation} \label{MPM}
\begin{tikzcd}[column sep=huge]
\M: \arrow[d,swap,shift left=-2pt,"\phi"] &[-45pt] \M(0) \arrow[r,"\M(0\leq 1)"] \arrow[d,"\phi_0"] &[-5pt] \M(1) \arrow[r,"\M(1\leq 2)"] \arrow[d,"\phi_1"] &[-5pt] \cdots \arrow[r,"\M(m- 1\leq m)"] &[5pt] \M(m) \arrow[d,"\phi_m"] \arrow[r,"\M(m\leq m+1)"] &[5pt] \cdots \\
\N: & \N(0) \arrow[r,"\N(0\leq 1)"] & \N(1) \arrow[r,"\N(1\leq 2)"] & \cdots \arrow[r,"\N(m- 1\leq m)"] & \N(m) \arrow[r,"\N(m\leq m+1)"] & \cdots
\end{tikzcd}
\end{equation}
is a natural transformation of functors $\mathbb{N}_0\rightarrow \mathbf{vec_{\field}}$. 
A persistence module $\N$ is a \emph{submodule} of $\M$ if $\N(i)\subseteq \M(i)$ and
$\N(i\leq i+1)$ is a restriction of $\M(i\leq i+1)$ to $\N(i)$ for all $i\geq 0$.
A submodule $\N$ of $\M$ is called its 
\emph{summand} 
if there is a submodule $\N'$ of $\M$ so that
$\M= \N\bigoplus \N'$. 
An \emph{interval module} $\mathbb{I}^{[\lbirth,\ldeath)}$, over an interval $[\lbirth,\ldeath)\subseteq\mathbb{N}_0$,
is a persistence module 
where
\begin{equation}\label{eqn:intvmod}
\mathbb{I}^{[\lbirth,\ldeath)}(i) \cong
\begin{cases}
\field & \text{if } i \in [\lbirth,\ldeath) \\
0 & \text{otherwise}
\end{cases}
\end{equation}
Moreover, $\mathbb{I}^{[\lbirth,\ldeath)}(i\leq j)$ is an isomorphism 
if $i,j\in[\lbirth,\ldeath)$
and is a zero map otherwise.
By the Theorems of Krull-Remak-Schmidt~\cite[Theorem 1]{azumaya1950corrections} and Gabriel~\cite[Chapter 2.2]{Gabriel72}, 
every persistence module of finite type
can be written as finite direct sum of interval modules
\begin{equation}\label{eqn:intv-decomp}
\M=\bigoplus_{i=1}^s \mathbb{I}^{[\lbirth_i,\ldeath_i)},
\end{equation}
where each $\mathbb{I}^{[\lbirth_i,\ldeath_i)}$ is a summand of $\M$.
The multiset of intervals, also called \emph{bars},
\[\{[\lbirth_i,\ldeath_i)\mid i=1,\ldots,s\}\] 
form the \emph{barcode} of $\M$.
We notice that
the decomposition of $\M$ as in Eq.~\eqref{eqn:intv-decomp} is typically written as 
$\M$ being \emph{isomorphic} (i.e., $\cong$) to a direct sum of interval modules.
However, if this is true, then we could always find a decomposition where each interval
module is a submodule, so that the pointwise vector spaces of $\M$ become equal
to a direct sum of vector spaces of the interval modules in the decomposition.
So, we turn ``$\cong$'' into an equality. 
Having the interval submodules as summands 
is essential for definitions that follow.

A persistence module is called \emph{free} if the barcode contains only infinite bars,
i.e., $\ldeath_i=\infty$ for each $i$ in Eq.~\eqref{eqn:intv-decomp}.
Suppose that a free module $\M$ has a decomposition 
$\M=\bigoplus_{i=1}^s \mathbb{I}^{[\lbirth_i,\infty)}$.
A non-zero $g_i\in\mathbb{I}^{[\lbirth_i,\infty)}(\lbirth_i)$ is called a \emph{generator}
of $\M$, where $\lbirth_i$ is called the \emph{grade} of $g_i$
and is denoted as $\gr(g_i)$.
A set of generators $\{g_1,\ldots,g_s\}$ for all
the infinite bars in a decomposition is called a \emph{generating set} of $\M$.
It is evident that each $x\in\M(j)$ can be written as
$x=\sum_{i=1}^t \lG_i\cdot \M(\lbirth_i\leq j)(g_i)$ 
for some scalars $\lG_1,\ldots,\lG_t\in \field$
where $\lbirth_1,\ldots,\lbirth_t\leq j$.

The well-known unique interval decompositions of $1$-parameter persistence modules~\cite{Oudot15,ZC05} implicitly imply
the following result.
\begin{proposition}
    A persistence module $\M$ is free iff every morphism $\M(i\leq j)$ is injective.
\label{prop:free}
\end{proposition}

\section{Linked barcode from filtration}
\label{sec:linkbarcode}

\subsection{Link module and barcode}\label{sec:persmod}

Usually, in TDA, we construct a persistence module from a given
simplicial filtration considering the homology functor in degree $p$ that maps each complex
in the filtration to its $p$-th homology group 
over $\field$ and each inclusion
to a linear map. Here, 
we consider other functors as well given by the \emph{chain}, \emph{cycle}, and \emph{boundary} spaces
of simplicial complexes.

A \emph{simplex-wise filtration} is a finite sequence of nested (finite) simplicial complexes. Using
the notation $\mathbf{simp}$ for the category of finite simplicial
complexes with morphisms being inclusions, a filtration $\Fcal$ can be seen as a functor
$\Fcal: \mathbb{N}_0\rightarrow \mathbf{simp}$ of finite type
where there exists an $m\geq 0$ so that $\Fcal(i\leq i+1)$ is identity for every $i\geq m$.
We write a simplex-wise filtration $\Fcal$ explicitly as:
\begin{equation}
\Fcal: K_0\inctosp{\sigma_1}K_1\inctosp{\sigma_2}\cdots\inctosp{\sigma_{m}} K_m\mbox{, where } 
K_{i}\setminus K_{i-1} =\{\sigma_{i}\} \mbox{ for } 0< i\leq m.
\label{eq:filtration}
\end{equation}

Given a simplicial complex $K$,
let $C_p(K)$, $Z_p(K)$, and $B_p(K)$ denote the groups of $p$-chains, $p$-cycles, and $p$-boundaries
respectively with coefficients in $\field$. These groups form nested $\field$-vector spaces
$B_p(K)\subseteq Z_p(K)\subseteq C_p(K)$. An inclusion
$K_i\inctosp{} K_{i+1}$ induces \emph{injective} linear maps 
which can also be considered as chain inclusions
in the respective chain, cycle, and boundary
spaces: $C_p(K_i)\incto C_p(K_{i+1})$, $Z_p(K_i)\incto Z_p(K_{i+1})$, and
$B_p(K_i)\incto B_p(K_{i+1})$.
A filtration $\Fcal$ as in~\eqref{eq:filtration} induces different persistence modules, namely, $p$-chain module $\C_p:=\C_p(\Fcal)$, and
$p$-cycle module $\Z_p:=\Z_p(\Fcal)$, and $p$-boundary module $\B_p:=\B_p(\Fcal)$ as: 
\begin{align}
\label{eqn:C}&\C_p: C_p(K_0)\incto C_p(K_1)\incto \cdots\incto  C_p(K_m),\\
\label{eqn:Z}&\Z_p: Z_p(K_0)\incto Z_p(K_1)\incto \cdots\incto  Z_p(K_m),\\
\label{eqn:B}&\B_p: B_p(K_0)\incto B_p(K_1)\incto \cdots\incto  B_p(K_m).
\end{align}

Furthermore, one has the following \emph{homology} module
formed by quotienting cycle spaces with boundaries and considering morphisms 
between them induced
by inclusions
\begin{equation*}
\Hm_p: Z_p(K_0)/B_p(K_0)\rightarrowsp{}Z_p(K_1)/B_p(K_1)\rightarrowsp{}\cdots\rightarrowsp{} Z_p(K_m)/B_p(K_m).
\end{equation*}
Above, each $Z_p(K_i)/B_p(K_i)$ is the $p$-th homology group
of $K_i$
and is  denoted
as $H_p(K_i)$.
Moreover, let $H(K_i):=\bigoplus_{p}H_p(K_i)$ be the homology group
over all dimensions.
Our main aim is to connect the homology module with another module which we call
the \emph{link} module. Next, we prepare for defining this link module beginning with the following result \cite[pp. 121]{EH10}.
\begin{proposition}
For each inclusion $K_{i-1}\inctosp{\sG_i}K_{i}$,
one has that either $\dim(H(K_{i}))=\dim(H(K_{i-1}))+1$
or  $\dim(H(K_{i}))=\dim(H(K_{i-1}))-1$.
\end{proposition}
\begin{definition}\label{dfn:posneg}
Call a simplex $\sG_i$ in $\Fcal$ \emph{positive} if 
$\dim(H(K_{i}))=\dim(H(K_{i-1}))+1$
and call it \emph{negative} otherwise.
Let $\sset_p(K_i)$ denote the set of $p$-simplices
of $K_i$ 
and $\pset_p(K_i)\subseteq \sset_p(K_i)$  the subset of positive $p$-simplices.
Define a subspace $X_p(K_i)\subseteq C_p(K_i)$
to consist of all chains in $C_p(K_i)$ supported over 
$\pset_p(K_i)$
and 
$Y_p(K_i)\subseteq C_p(K_i)$
to consist of all chains in $C_p(K_i)$ supported over 
$\sset_p(K_i)\setminus \pset_p(K_i)$.
We then have the following modules $\X_p$, $\Y_p$
over the chain spaces $X_p(K_i)$, $Y_p(K_i)$
connected by chain inclusions:
\begin{align*}
&\X_p\,:\, X_p(K_0) \incto X_p(K_1)\incto\cdots\incto X_p(K_m),\\
&\Y_p\,:\, Y_p(K_0) \incto Y_p(K_1)\incto\cdots\incto Y_p(K_m).
\end{align*}
\end{definition}
\begin{proposition}
    (i) $\forall i$, $C_p(K_i)=X_p(K_i)\oplus Y_p(K_i)$, and (ii) $C_p(K_i)=Z_p(K_i)\oplus Y_p(K_i)$.
    \label{prop:point-split}
\end{proposition}
\begin{proof}
   (i) follows straightforwardly because the set of $p$-simplices in $K_i$ constitutes a basis for $C_p(K_i)$ and it can be partitioned into positive and negative simplices forming
    a basis for $X_p(K_i)$ and $Y_p(K_i)$ respectively.
    
    (ii) follows because there is a bijection from a basis of $X_p(K_i)$ to a basis of $Z_p(K_i)$
    as follows: for every positive simplex $\sigma$ choose a  $p$-cycle $z_\sigma$ in $Z_p(K_i)$ which has $\sigma$ as the  last $p$-simplex in the filtration among its supporting simplices. It is well known that this set of $p$-cycles forms a basis of 
    $Z_p(K_i)$. The bijective mapping $\sigma\mapsto z_\sigma$ induces an isomorphism $X_p(K_i)\cong Z_p(K_i)$ both being subspaces of $C_p(K_i)$.
    Furthermore, $Z_p(K_i)\cap Y_p(K_i)=\{0\}$ by definition, completing the proof that $C_p(K_i)=Z_p(K_i)\oplus Y_p(K_i)$.
\end{proof}
\begin{proposition}
The following holds for each $p$:
    \begin{enumerate}
        \item[1.] $\C_p$, $\Z_p$, $\X_p$, and $\Y_p$  are free;
        \item[2.] $\C_p=\Z_p\oplus \Y_p$.
    \end{enumerate}
\label{prop:basic}
\end{proposition}
\begin{proof}
(1) follows because the structural homomorphisms $\C_p(i\leq j)$,
$\Z_p(i\leq j)$, $\X_p(i\leq j)$, and $\Y_p(i\leq j)$ are injective.
(2) follows because both $\Z_p$ and $\Y_p$ are submodules of
$\C_p$ and they are summands point-wise, $C_p(K_i)=Z_p(K_i)\oplus Y_p(K_i)$ (Proposition~\ref{prop:point-split}) for $0\leq i\leq m$. 
\end{proof}

It follows from Proposition~\ref{prop:point-split} that 
    for any $z_p\in Z_p(K_i)$, one has a unique sum
    $z_p=x_p + y_p$ where $x_p\in X_p(K_i)$ and
    $y_p\in Y_p(K_i)$.
Then, we have a well defined
morphism $\chi:\Z_p\to \X_p\oplus \Y_p$ where each point-wise $\chi_i: Z_p(K_i)\rightarrow X_p(K_i)\oplus Y_p(K_i)$
is defined as:
\begin{equation}\label{eqn:split}
\chi_i:z\mapsto(x,y)\text{ so that }z=x+y.
\end{equation}
We furthermore consider a ``projection'' morphism $\phi:\Z_p\to\Y_p$
\begin{equation} 
\label{eqn:phi}
\begin{tikzcd}[column sep=huge]
\Z_p: \arrow[d,swap,shift left=-2pt,"\phi"] &[-45pt] Z_p(K_0) \arrow[hookrightarrow,r,"\iG"] \arrow[d,"\phi_0"] &[-5pt]  Z_p(K_1)  \arrow[hookrightarrow,r,"\iG"] \arrow[d,"\phi_1"] &[-5pt] \cdots \arrow[hookrightarrow,r,"\iG"] &[5pt]  Z_p(K_m) \arrow[d,"\phi_m"] \\
\Y_p: & Y_p(K_0) \arrow[hookrightarrow,r,"\iG"] & Y_p(K_1) \arrow[hookrightarrow,r,"\iG"] & \cdots \arrow[hookrightarrow,r,"\iG"] & Y_p(K_m) 
\end{tikzcd}
\end{equation}
Above,
each point-wise $\phi_i: Z_p(K_i)\rightarrow Y_p(K_i)$
is defined as
\begin{equation}\label{eqn:proj}
\phi_i:z\xmapsto{\;\chi_i\;}(x,y)\xmapsto{\;\pi_i\;} y,
\end{equation}
where $\pi_i$ maps a tuple $(x,y)$ to its second component $y$.
In words, $\phi_i$ maps a $p$-cycle $z$ in $K_i$ to the $p$-chain $y$ in $K_i$ 
where $y$ morphs into $z$ after the addition of some positive $p$-simplices in the chain $x$.
It is easy to see that 
the $\phi_i$'s commute with the horizontal inclusions $\iG$'s in the diagram in \Cref{eqn:phi},
and hence $\phi$ is a well-defined morphism
between persistence modules.\begin{definition}[$p$-link module and barcode]\label{dfn:lmod}
We define the \emph{$p$-link module} as the cokernel module $\Lsf_p:=\coker(\phi)$, and explicitly
\begin{equation} 
\label{eqn:lmod}
\Lsf_p\,:\, Y_p(K_0)/\img(\phi_0) \rightarrowsp{} Y_p(K_1)/\img(\phi_1)\rightarrowsp{}\cdots\rightarrowsp{} Y_p(K_m)/\img(\phi_m).
\end{equation} 
In \Cref{eqn:lmod}, the connecting morphisms are
induced by inclusions, i.e.\ for $y\in Y_p(K_i)$,
\[y+\img(\phi_i)\mapsto \iG(y)+\img(\phi_{i+1}),\]
which is well defined {\rm(}i.e., $w\in\img(\phi_i)$ implies 
$\iG(w)\in\img(\phi_{i+1})${\rm)} by the commutativity of the diagram in
\Cref{eqn:phi}.
We call the barcode of $\Lsf_p$ the
\emph{$p$-link barcode} and its bars the \emph{$p$-link bars}.
\end{definition}

Observe that the cycle module $\Z_0$ coincides with
the chain module $\C_0$ and hence the module $\Y_0$ is trivial. 
As a result, $\Lsf_0=0$ 
and we focus on the $p$-link modules $\Lsf_p$ for $p>0$.
We call them $p$-link modules
for the given filtration $\Fcal$ because 
the finite bars in the barcode of $\Lsf_p$ connect
the bars in the barcode of the homology module $\Hm_{p-1}$ with the bars in the barcode
of the homology module $\Hm_p$. Furthermore, each infinite bar in the barcode of $\Lsf_p$
begins with a right end point of a bar in the barcode of $\Hm_{p-1}$. 
See \Cref{thm:link} for a formal statement of these facts.

\cancel{
Next, consider the set of chains $W_p(K_i)$ as defined below:
\begin{equation}
W_p(K_i)=\{y_p(K_i)\in Y_p(K_i)\,|\, y_p(K_i)+x_p(K_i)\in Z_p(K_i) \mbox{ for some } x_p(K_i)\in X_p(K_i)\}.
\label{eq:w}
\end{equation}
In words, $W_p(K_i)$ are those $p$-chains in $C_p(K_i)$ that do not form $p$-cycles
under $\field$-additions among themselves but do so with an addition from chains formed
by last $p$-simplices that complete $p$-cycles.
\begin{proposition}
$W_p(K_i)$ is a subspace of $Y_p(K_i)$.
\end{proposition}
\begin{proof}
    We only need to show that addition of any two elements $y_p(K_i)$ and $y'_p(K_i)$
    in $W_p(K_i)$ creates an element in $W_p(K_i)$ which implies that $W_p(K_i)$
    is a subgroup of $Y_p(K_i)$...to be completed. 
\end{proof}
Now consider the persistence module
$$
\W_p: W_p(K_0)\rightarrowsp{}W_p(K_1)\rightarrowsp{}\ldots\rightarrowsp{} W_p(K_m)
$$
where the maps $W_p(K_i)\rightarrowsp{}W_p(K_{i+1})$, induced by inclusions, are injective.
\begin{proposition}
    $\W_p$ is a free summand of $\Y_p$.
    \label{prop:wy}
\end{proposition}
\begin{proof}
  The fact that it is a summand of $\Y_p$ follows from the fact that point-wise each $W_p(K_i)$ is a summand (subspace) of $Y_p(K_i)$ and the morphisms of $\Y_p$ restricts to $\W_p$ because both are induced by inclusions. Since a summand of a free module is necessarily free, we have that $\W_p$ is free.
\end{proof}
The following corollary is immediate from Proposition~\ref{prop:basic} and Proposition~\ref{prop:wy}.
\begin{corollary}
    $\C_p=\Z_p\oplus \W_p\oplus \U_p$ for some summand $\U_p$.
    \label{cor:cw}
\end{corollary}
}

\subsection{Presentations and barcode}
\label{sec:linkedbarcode}

We now introduce the notion of presentation of a module, 
which helps prove important facts (e.g., \Cref{thm:link})
regarding the link barcode and also helps justify
the algorithm for computing the link barcode  (see \Cref{sec:alg}).

\begin{definition}[Presentation of module] 
Let $\M$ be a persistence module of finite type. 
A presentation of $\M$ is an exact sequence of the form 
$
\Ps_1\rightarrowsp{f} \Ps_0\rightarrowsp{} \M\rightarrowsp{} 0,
$
where $\Ps_0$ and $\Ps_1$ are free modules of finite type. The presentation is called \emph{minimal} if 
$\ker f=0$ and 
no two generators 
$g_0$ of $\Ps_0$ and $g_1$ of $\Ps_1$ with the same grade $j$ satisfy that $f_j(g_1)=g_0$,
where $f_j:\Ps_1(j)\to\Ps_0(j)$ is the linear map at index $j$
in $f$.
\label{def:pre_module}
\end{definition}

The following fact is well known, e.g., it can be derived from results in~\cite{ZC05}.
\begin{proposition}
 Let $\Ps_1\rightarrowsp{f} \Ps_0\rightarrowsp{} \M\rightarrow{} 0$ be a presentation of $\M$. Then,
 for every bar $[\lbirth,\ldeath)$ in the barcode of $\M$, $\lbirth$ is the grade of a generator of $P_0$ and
 $\ldeath$ is either $\infty$ or a grade of a generator of $\Ps_1$.
 Conversely, if the presentation is minimal, for every grade $\lbirth$ of a generator of $\Ps_0$, there is a bar $[\lbirth,-)$ in the barcode of $\M$ and for every grade $\ldeath$ of a generator of $\Ps_1$, there is a bar $[-,\ldeath)$ in the barcode of $\M$.
 \label{prop:grades}
\end{proposition}

\begin{definition}[Presentation matrix]
    Let $\M$ have a presentation as in Definition~\ref{def:pre_module} and let
    $\{g_1,\ldots, g_s\}$ and $\{h_1,\ldots,h_r\}$ be some generating sets
    for $\Ps_0$ and $\Ps_1$ respectively ordered by non-decreasing grades.
    The 
    morphism $f$ can be represented with a \emph{presentation matrix} 
    $[f]=(\alpha_{ij})_{s\times r}$,
    with $j$-th column and $i$-th row \emph{graded} with $\gr(h_j)$ and
    $\gr(g_i)$ respectively.
    Moreover, if $\Ps_0, \Ps_1\not= 0$,
    each column $(\alpha_{1j},\ldots,\alpha_{sj})^\tran$ of $[f]$ satisfies the following:
    Let $t$ be the greatest index $i$ among $1,\ldots,s$ with $\gr(g_i)\leq \gr(h_j)$. Then
    \[{f_{\gr(h_j)}(h_j)=\sum_{i=1}^t \alpha_{ij}\cdot\Ps_0(\gr(g_i)\leq\gr(h_j))(g_i).} \]
    Moreover, for $t<i\leq s$, we have $\alpha_{ij}=0$.
    We also call $[f]$ a \emph{graded matrix} of $\M$.
    Notice that
    when $\Ps_1=0$ $(\Ps_0=0)$, the matrix $[f]$ degenerates into an empty column (row) matrix that has
    only rows (columns) for every
    generator of $\Ps_0$ $(\Ps_1)$. This helps in interpreting this degenerate case
    in algorithms that we design later.
    \label{def:presentmatrix}
\end{definition}

The barcode of a module $\M$ can be computed from 
a graded matrix $[f]$ of $\M$. For this, as shown in~\cite{ZC05},
one simply needs to \emph{reduce} $[f]$ by left-to-right
column reductions as done for the well known persistence algorithm~\cite{EH10}. 
For a non-zero column $j$
of a graded matrix $A$,
let $\low_A(j)$ denote the row grade of the lowest non-zero entry in the column.
Essentially, this reduction produces a (reduced) matrix $R$ with unique 
$\low_R(j)$ for each non-zero column $j$ s.t.\ $R=[f]V$ for an upper triangular matrix $V$. 
Notice that when $\Ps_1=0$ or
    $\Ps_0=0$,
    $R$ is also taken to be a degenerate matrix.
The
following result from~\cite{ZC05} is well known. 
\begin{proposition}
    Let $\Ps_1\rightarrowsp{f} \Ps_0\rightarrowsp{} \mathsf{M}\rightarrow{} 0$ be a presentation (not necessarily minimal) of $\M$ and
    let $R=[f]V$ be a reduced matrix of the graded matrix $[f]$. Then:
    \begin{itemize}
        \item For each non-zero column $j$ of $R$, the pair $[\low_R(j),\gr(j))$ is a bar in the barcode of $\M$, where $\gr(j)$ is the grade of column $j$.
        \item For each row $i$ of $R$ where $\gr(i)\not = \low_R(j)$ for any non-zero column $j$, $[\gr(i),\infty)$ is
        a bar in the barcode of $\M$.
    \end{itemize}
    \label{prop:reduce}
\end{proposition}

\begin{theorem}
    For $p>0$, let $[\lbirth,\ldeath)$ be a $p$-link bar in the barcode of the $p$-link module $\Lsf_p$.
    Then, there is a bar $[-,\lbirth)$ in the barcode of the homology module $\Hm_{p-1}$ and there is a bar $[\ldeath,-)$ in the barcode of the homology
    module $\Hm_p$ if $\ldeath\not=\infty$.
    \label{thm:link}
\end{theorem}
\begin{proof}
    First, we derive a minimal presentation of the homology modules.
    It is well known that the following is a presentation of the homology module $\Hm_p$:
    \begin{equation*}
    \C_{p+1}\rightarrowsp{\partial_{p+1}}\Z_p\rightarrowsp{}\Hm_p\rightarrowsp{} 0
    \end{equation*}
    where $\partial_{p+1}$ denotes the point-wise boundary morphism $\partial_{p+1}(K_i): C_{p+1}(K_i)\rightarrowsp{} C_p(K_i)$. This presentation, however, is not minimal because
    $\ker\, \partial_{p+1}=\Z_{p+1}$ may not be zero. From Proposition~\ref{prop:basic}, we have $\C_{p+1}=\Z_{p+1}\bigoplus \Y_{p+1}$,
    and hence $\partial_{p+1}(Y_{p+1}(K_i))=\partial_{p+1}(C_{p+1}(K_i))$
    for each $i$. We then have the following presentation
    \begin{equation}
    \Y_{p+1}\rightarrowsp{\partial_{p+1}}\Z_p\rightarrowsp{}\Hm_p\rightarrowsp{} 0
    \label{eq:minpresent}
    \end{equation}
    where $\ker\, \partial_{p+1}=0$. Furthermore, since the input filtration $\Fcal$ is simplex-wise,  no
    generator of $\C_{p+1}$ and hence of $\Y_{p+1}$ can coincide with any generator of
    $\Z_p$ in grades. It follows that the presentation in Eq.~\eqref{eq:minpresent} is minimal.
    
    Next, consider the presentation of the module $\Lsf_p$:   
    \begin{equation}
        \Z_p\rightarrowsp{\phi}\Y_p \rightarrowsp{} \Lsf_p\rightarrowsp{} 0.
        \label{eq:linkpresent}
    \end{equation}
    Applying Proposition~\ref{prop:grades}, we conclude that {every bar $[\lbirth,\ldeath)$ in the barcode of $\Lsf_p$ with
    $\ldeath\not=\infty$ has $\ldeath$ as the grade of a generator for $\Z_p$}.
    Then,
    since Eq.~\eqref{eq:minpresent} is a minimal presentation of the module $\Hm_p$, there is a bar $[\ldeath,-)$ in the barcode of $\Hm_p$ according to Proposition~\ref{prop:grades}. 

    Also, according to Proposition~\ref{prop:grades}, the grade $\lbirth$ coincides with the grade
    of a generator for $\Y_p$. Then, again applying Proposition~\ref{prop:grades} to the minimal presentation in Eq.~\eqref{eq:minpresent}
    adapted to $\Hm_{p-1}$,
    we get that there is a bar $[-,\lbirth)$ in the barcode of $\Hm_{p-1}$.
\end{proof}

\begin{remark}
A link bar is the lifespan of an evolving chain during which it remains 
as a non-cycle. 
Equivalently, it measures the time gap between the destruction of a 
$(p-1)$-cycle when the chain first appears and the creation of a $p$-cycle that 
the evolving chain grows into. The linked barcode records a coupling between standard homology barcodes in successive dimensions. For example, consider the following two filtrations with four vertices. First, the vertices $a,b,c,d$ appear, in that order, followed by: 
\begin{align*}
    \mathcal{F}_1 :  e_1=(a,b), e_2=(b,c), e_3=(c,d), e_4=(b,d), e_5=(a,c)\\
    \mathcal{F}_2: e_1=(a,b), e_2=(b,c), e_3=(c,d), e_4=(a,c), e_5=(b,d) 
\end{align*}
Both filtrations, $\mathcal{F}_1$, and $\mathcal{F}_2$, have the same persistence barcodes. However, in the filtration $\mathcal{F}_1$, the cycle $e_1+e_2+e_5$ takes longer to form. This corresponds to the bar $(e_2,e_5)$ in the linked barcode of $\mathcal F_1$. In contrast, the same cycle corresponds to a shorter bar $(e_2,e_4)$ in the linked barcode of $\mathcal{F}_2$; this shift is captured by the linked barcode.
\end{remark}
\cancel{
\subsubsection{Linked barcode} 
\label{sec:linkedbar}
The barcodes for homology modules $\Hm_{p-1}$ and $\Hm_p$ can be connected with the help of
$p$-link barcodes for all $p>0$ to generate an elongated barcode we call the \emph{linked
barcode} for the input filtration $\Fcal$.
Let 
\[x[p,q]=(x_1^{p},x_2^{p},x_1^{p+1},x_2^{p+1}\ldots,x_1^{q},x_2^{q})\]
be an
ordered sequence where $[x_1^{p+i},x_2^{p+i})$ is a bar in the barcode of $\Hm_{p+i}$ for $0\leq i\leq q-p$. We say
$x[p,q]$ is a \emph{linked bar} given that the following are true: (i) there is no bar $[-,x_1^{p})$ in the barcode of $\Lsf_p$;
(ii) there is a $(p+i)$-link bar $[x_2^{p+i},x_1^{p+i+1})$ in the barcode of $\Lsf_{p+i}$ for $0\leq i\leq (q-p)$;
(iii) $[x_1^q,x_2^q)$ is either an infinite bar of $\Hm_q$, i.e., $x_2^q=\infty$, or
$[x_1^q,x_2^q)$ is a finite bar and
there is an infinite bar $[x_2^q,\infty)$ in the barcode of $\Lsf_{q}$.
The collection of
all linked bars constitutes what we call the \emph{linked barcode} \gilberto{should we call it link barcode or linked barcodeç? maybe even linking barcode, I don't think is a linked barcode but a linking barcode, naming it link barcode may be better, we need to fix the consistency through the paper}  of $\Fcal$. Figure~\ref{fig:linkedbar} shows five linked bars in the linked barcode of the illustrated filtration.}

\section{Algorithm}\label{sec:alg}
Theorem~\ref{thm:link} characterizes the $p$-link bars 
as links between the bars of the barcodes of $\Hm_{p-1}$ and $\Hm_p$. Now, we 
present an algorithm to compute these $p$-link bars. 
As Proposition~\ref{prop:reduce} suggests, the barcode of $\Lsf_p$ can then 
be computed by reducing the matrix $[\phi]$ obtained
from its presentation in Eq.~\eqref{eq:linkpresent}. Therefore, our main task is to compute
the matrix $[\phi]$ from an input filtration $\Fcal$ as in Eq.~\eqref{eq:filtration}.

To compute $[\phi]$, recall the following factoring of $\phi$
as in \Cref{sec:persmod}:
\[\phi: \Z_p\rightarrowsp{\chi} \X_p\oplus\Y_p \rightarrowsp{\pi} \Y_p.\] 
Specifically, the morphisms
are defined point-wise as: $\chi_i(z)=(x,y)$ and $\pi_i(x,y)=y$ where
$z=x+y$ for $z\in Z_p(K_i)$, $x\in X_p(K_i)$, and $y\in Y_p(K_i)$.
Therefore, we only need to compute the
matrices $[\chi]$, $[\pi]$ and then reduce $[\phi]=[\pi][\chi]$ to get the $p$-link barcode of $\Lsf_p$.

From a simplex-wise filtration $\Fcal$, one can derive
a 
graded boundary matrix
$[\partial_p]$
with bases being \emph{elementary} chains each consisting of a single simplex.
Moreover,
the columns and rows of $[\partial_p]$ are graded 
with the grades of the generators 
they represent, i.e., the column (row resp.) 
corresponding to $\sigma_j$ has grade $j$. 
Let $\col_A(j)$ ($\row_A(j)$ resp.)\ denote the column (row resp.)\ with grade $j$ in a matrix $A$,
we have that $\col_{[\partial_p]}(j)$ has $\alpha_{ij}$ in its row with grade $i$ if $\sigma_i$ appears
with coefficient $\alpha_{ij}\in \field$ in $\partial_p \sigma_j$.
Suppose we reduce $[\partial_p]$ with left-to-right column reductions (see e.g. \cite{EH10})
to obtain the reduced matrix
$R=[\partial_p]V$,  where $V$ is a (square) upper triangular 
matrix whose rows and columns have the grades of the columns of $[\partial_p]$.
The following fact is well known, see \cite[Corollary 4.1]{ZC05}.
\begin{fact}
    The column $\col_V(j)$ represents a $p$-cycle $z$ in $\Z_p$ generated at 
    grade $j$ if
    $\col_R(j)=0$. 
    Conversely,
    $\col_V(j)$ represents a $p$-chain
    $y$ in $\Y_p$ generated at grade $j$ 
    if $\col_R(j)\neq 0$.
    Furthermore, the set of $p$-cycles $z_1,\ldots, z_{s}$ and 
    the set of $p$-chains $y_1,\ldots, y_{r}$, represented by all columns of $V$,
    constitute a basis for $\Z_p$ and $\Y_p$ respectively. 
\end{fact}

We call the columns of $V$ representing a basis of $\Z_p$
and $\Y_p$
the \emph{$\Z$-columns} and \emph{$\Y$-columns} respectively.
The rows of $V$ similarly
get partitioned into two sets, \emph{$\X$-rows} and \emph{$\Y$-rows}, whose grades match
 those of the $\Z$-columns and $\Y$-columns respectively. 
Effectively, we get a basis $\{\sigma_{i_1},\ldots,\sigma_{i_s}\}$
for $\X_p$, 
where $i_1,\ldots, i_s$ are the grades of $\X$-rows of $V$.
Similarly, we get a basis for $\Y_p$ from the $\Y$-rows of $V$.
Then, the submatrix
$V^\chi$ of $V$ consisting of $\Z$-columns and all rows represents the morphism 
$\Z_p\rightarrowsp{\chi} \X_p\oplus \Y_p$. Let $I^\pi$
be a $r\times (r+s)$
truncated identity matrix where $V$ has
{$s$} $\X$-rows and {$r$} $\Y$-rows. 
The matrix $I^\pi$
 represents the projection
$\X_p\oplus \Y_p\rightarrowsp{\pi}\Y_p$. 

The product matrix $V'=I^\pi V^\chi$ therefore
represents the morphism $\Z_p\rightarrowsp{\phi=\pi\circ\chi} \Y_p$.
However, for computational efficiency, we do not explicitly multiply the
two matrices and instead obtain $V'$
by
simply deleting the $\X$-rows and $\Y$-columns of $V$.  In doing so, 
rows and columns of $V'$ inherit the grades of those of $V$.
Also notice that
if the column or row space of
$V'$ become zero, then it degenerates to an empty column or row matrix
as described in Definition~\ref{def:presentmatrix}. 
We then reduce $V'$ to obtain the barcode of $\Lsf_p$.
Specifically, let $R'$ be the reduced version of $V'$.
Applying Proposition~\ref{prop:reduce}, we get:
\begin{proposition}\label{prop:link-reduce}
    (i) For each non-zero column $j$ of $R'$,
    $[\low_{R'}(j),\gr(j))$ is a finite bar in the barcode of $\Lsf_p:=\coker(\pi\circ\chi)$.
    (ii) For each row $i$ of $R'$ where $\gr(i)\not = \low_{R'}(j)$ for any non-zero column $j$, $[\gr(i),\infty)$ is
        an infinite bar in the barcode of $\Lsf_p$.
\label{prop:bars}
\end{proposition}

\begin{algorithm}[!h]\label{alg:link}
    \caption{{\sc{LinkedBar}($\Fcal$)}}
    \KwData{An $m\times m$ matrix $D$ of a graded boundary matrix computed from $\Fcal$} 
    \KwResult{Links of $\Fcal$}
    Initialize $V$ as an $m\times m$ identity matrix\;
    \For{$j:=1$ to $m$}{ 
        \While{$\exists i<j \mbox{ so that } \low_D(j)=\low_D(i)$}{
             $\col_D(j):=\col_D(j)+\col_D(i)$;
            $\col_V(j):=\col_V(j)+\col_V(i)$\;
        }
    }
    \For{every column $j$ of $V$ with $\col_D(j)$ non-zero}{
    $V:= V\setminus \col_V(j)$\;
    }
    \For{every row $i$ of $V$ with $\col_D(i)$ being zero}{
    $V:= V\setminus \row_V(i)$;  
    }
    Rename $V$ as $V'$\;
    \For{$j:=1$ to $\#\col_{V'}$}{ 
        \While{$\exists i<j \mbox{ so that } \low_{V'}(j)=\low_{V'}(i)$}{
             $\col_{V'}(j):=\col_{V'}(j)+\col_{V'}(i)$;
        }
    }
\Return{barcode from reduced $V'$ according to Proposition~\ref{prop:bars}}
\end{algorithm}

For computation, we can consider all dimensions of homology together. Thus,
instead of considering the boundary matrix $[\partial_p]$ separately for each $p>0$,
we consider the total boundary matrix $D=[\oplus_p\partial_p]$.
Given an input filtration $\Fcal$ as in~\eqref{eq:filtration},
the algorithm {\sc LinkedBar}
computes the $p$-link barcode for all $0< p \leq k$,
where $k$ is the dimension of the complex $K:=K_m$
and is assumed to be non-zero. Although the algorithm works 
with arbitrary field $\Bbbk$, we write the algorithm assuming $\Bbbk=\mathbb{Z}_2$ for
simplicity.
We call the set of $p$-link barcodes for $0< p \leq k$ the \emph{links} of $\Fcal$.
It is evident that the algorithm takes $O(m^3)$ time 
with the running time
dominated by left-to-right
column reductions of the two $m\times m$ matrices.


\section{Algorithm for graph filtration}
Our next goal is to design a near-linear time algorithm for the case when $\Fcal$ is a 
filtration of a graph. It is known that the standard persistence barcode for a graph filtration
can be computed in $O(n\alpha(n))$ time where $\alpha(n)$ is the extremely slowly growing inverse Ackermann's function. Even for the zigzag persistence, graphs admit better algorithm with
time complexity of $O(n\log n)$~\cite{DHP23}. We show that we can compute the linked barcode 
for a graph filtration in $O(n\log n)$ time. To prepare for the presentation of this algorithm,
we first develop a slightly different formulation of the linked barcode in the section to follow.

Assume that we have a filtration $\Fcal: G_0\inctosp{}G_1\inctosp{}\cdots\inctosp{} G_m$ of a graph $G=G_m$. Each vertex and edge in $G$ is
assigned a weight equal to its index in $\Fcal$ which is also its grade as a generator
in the graded chain module. 
\cancel{
By proposition~\ref{prop:linkedbar},
we have $\Lsf_1=\widehat{\Lsf}_1\oplus \U_1$. 
By Eq.~\eqref{eq:linksplit} and Proposition~\ref{prop:grades}, every bar $[a,b)$ in the barcode of $\Lsf_1$ has
$a$ coinciding with the grade of a generator in $\Y_1=\W_1\oplus\U_1$ and $b$ coinciding
with a generator in $\Z_1$ if $b\not=\infty$. Let $N$ and $P$ denote the set of
grades of generators of $\Y_1$ and $\U_1$ respectively.
}
For the next definition, recall that $X_1(G_i)$ is the chain space spanned by the last edges (by filtration order) that create
$1$-cycles in $G_i$. Consider the graded module
$\X_1: X_1(G_0)\rightarrowsp{}X_1(G_1)\rightarrowsp{}\cdots\rightarrowsp{} X_1(G_m)$.
\begin{definition}
A $1$-chain $y\in Y_1(G)$ is
called a \emph{representative} for a bar $[a,b)$ with $b\not=\infty$ in the barcode of $\Lsf_1$ if
(i) $y\in Y_1(G_a)$ is generated at grade $a$ and (ii) $y + x$ becomes a $1$-cycle with
the addition of a $1$-chain $x\in X_1(G_b)$ at grade $b$. If $b=\infty$, requirement (ii)
is dropped.
\end{definition}
The following result is a rephrasing of Proposition 9 in~\cite{DH21}.
\begin{proposition}
    Let $P$ and $N$ denote the set of grades of generators of $\Y_1$ and $\X_1$ respectively.
    Let $\mu: P\rightarrow (N\sqcup P)$ be an injection. Then, every interval
    $[a,\mu(a))$ with $\mu(a)\not=a$ is a finite bar and every interval $[a,\infty)$ with $a=\mu(a)$ is an infinite bar in the barcode of $\Lsf_1$ if and only
    if there is a representative for each of them.
    \label{prop:graphbarcode}
\end{proposition}
\begin{algorithm}[!t]
    \caption{{\sc{GraphLink}($\Fcal$)}}
    \KwData{A filtration $\Fcal: G_0\inctosp{}G_1\inctosp{}\cdots\inctosp{} G_m$ of a graph $G=G_m$ with $n$ vertices and edges weighted by its grade $\gr(\cdot)$ in $\Fcal$}
    \KwResult{The barcode of the $1$-link module $\Lsf_1$}
    $X_1:=\emptyset$; $Y_1:=\emptyset$\;
    Initialize a union-find structure with all vertices as singletons\;
    \tcp{Partition edges into $X_1$ and $Y_1$}
    \For{each edge $e=(u,v)$ in filtration order}{
        \eIf{$\textsc{Find}(u)=\textsc{Find}(v)$}{
            $X_1:=X_1\cup\{e\}$\;
        }{
            $Y_1:=Y_1\cup\{e\}$\;
            $\textsc{Union}(u,v)$\;
        }
    }
 
    Initialize a Link-Cut tree $T$ with all vertices and no edges\;
    \For{each edge $e\in Y_1$ in filtration order}{
        Insert $e$ into $T$\;
    }
    $B:=\emptyset$ \tcp*{set of finite bars}
    \For{each edge $e'=(u',v')\in X_1$ in filtration order}{
        $e:=$ edge of maximum grade on the path between $u'$ and $v'$ in $T$ that is also on $Y_1$. Set to empty if no such edge exists\;  
        \If{$e$ is not empty}{
            $B:=B\cup\{[\gr(e),\gr(e'))\}$\;
            Delete $e$ from $T$; insert $e'$ into $T$\;
        }
    }
    \For{each unpaired edge $e\in Y_1$}{
        $B:=B\cup\{[\gr(e),\infty)\}$\;
    }
    \Return $B$
\end{algorithm}
\para{Algorithm {\sc GraphLink}.}
The above proposition suggests an algorithm for computing the barcode of $\Lsf_1$. First, we compute
a set of generators for $\X_1$ and $\Y_1$ by running a minimum spanning tree algorithm
on $G$ with the edge weights equal to their filtration order. The algorithm maintains
a spanning forest. When an edge $e=(u,v)$ is added, it checks if the two vertices $u$ and $v$ are
already connected and, if so, it puts $e$ in a set $X_1$ and in a set $Y_1$ otherwise. These
sets constitute a basis for free modules $\X_1$ and $\Y_1$ respectively. 
By union-find data structure
this can be done in $O(n\alpha(n))$ time if $G$ has $n$ vertices and edges. 

Next, we compute bars of $\Lsf_1$ following Proposition~\ref{prop:graphbarcode}.
The algorithm computes the spanning
forest $T$ comprising the edges in $\Y_1$. Then, it brings in edges 
from $\X_1$ one by one in the filtration order and modifies the spanning forest $T$ while determining the pairs of edges
$(e,e')$ where $e\in \Y_1$ and $e'\in \X_1$ so that $[\gr(e),\gr(e'))$ constitutes
a finite bar in the barcode of $\Lsf_1$. At any generic step, let $e'\in \X_1$ be brought in whose
two endpoints are $u'$ and $v'$.
The edge $e'$ along with the simple path $h$ (not repeating vertices) that connects
$u'$ and $v'$ in $T$ creates a $1$-cycle. If $h$ has any edge from $\Y_1$, we generate a pair $(e,e')$ where $e$ is the edge with the largest grade (weight) which

is in both $\Y_1$ and $h$. In this case, we also
replace $e$ with $e'$ in $T$ keeping it as a spanning forest. 
\begin{proposition}
    $[\gr(e),\gr(e'))$ constitutes
a finite bar in the barcode of $\Lsf_1$.
\label{prop:graphbar}
\end{proposition}
\begin{proof}
    First, we observe that if $(e,e')$ is a pair computed by the algorithm {\sc GraphLink}, then
    $\gr(e)$ and $\gr(e')$ are in the set of grades $P$ and $N$ respectively because
    $e$ creates a new chain in $Y_1$ and $e'$ creates a new chain in $X_1$. 
    Define the map $\mu: P\rightarrow (N\sqcup P)$ by $\gr(e)\mapsto \gr(e')$ if
    $(e,e')$ is a pair and by $\gr(e)\mapsto \gr(e)$ otherwise. Observe that
    every pair $(e,e')$ generated by the algorithm is disjoint in both components because any time
    an edge $e\in Y_1$ is output in a pair, it is deleted from the current spanning tree $T$ and
    every time only a new edge $e'\in X_1$ is brought in. This ensures that the map $\mu$ is an injection.
    
    Consider any pair $(e,e')$ generated by {\sc GraphLink} and let
    $h$ be the path connecting the two end points of $e'$. Let the
    $1$-chain $y\in Y_1$ be given by the sum of edges in $h$. It is
    a representative for the interval $[\gr(e),\gr(e'))$ because the chain
    $y$ is created with the addition of the last edge $e\in Y_1$ in the filtration and $y+e'$ becomes
    a cycle with the addition of the last edge $e'\in X_1$. For each edge $e$ in
    the rest of the edges
    in $Y_1$ that are not paired, the $1$-chain $y\in Y_1$ created by the addition of $e$ constitutes a representative of the interval $[\gr(e),\infty)$. The conclusion follows from Proposition~\ref{prop:graphbarcode}.
\end{proof}
To implement the algorithm {\sc GraphLink}, we use the Link-Cut tree~\cite{ST81} data structure.
It can support the following operations in $O(\log n)$ time to maintain a forest $T$ of size $O(n)$:
(i) insert/delete a vertex or edge in $T$, (ii) determine an edge with the maximum weight
on the unique path in $T$ connecting two vertices. Precisely, these are the two operations
needed to implement {\sc GraphLink}. If the filtration $\Fcal$ inserts at most $n$ vertices and edges, initially the minimum spanning forest $T$ with all edges in $Y_1$ is built in $O(n\log n)$ time. Then, maintaining $T$ over all $O(n)$ insertions and deletions of edges take at most
$O(n\log n)$ time.
\begin{theorem}
    Given a filtration $\Fcal$ of a graph with $n$ insertions, the links
    of $\Fcal$ can be computed in $O(n\log n)$ time.
\end{theorem}

\section{Stabilizing link barcode}
The link barcode as we defined is not stable with respect to perturbations in the
filtration. For example, consider the two filtrations shown in Figure~\ref{fig:instability}. 
The reader can verify that we get a link bar $e_3-e_{10}$ in the left
filtration. The right filtration is obtained from the left one
with a single transposition of two consecutive edges.
This change turns the edge $e_{9}$ (which was $e_{8}$ before the transposition)
positive. As a result, the link bar $e_3-e_{10}$ vanishes because $e_{10}$
now appears in a cycle of only positive edges.

\begin{figure}[htbp]
\centerline{\includegraphics[width=0.5\textwidth]{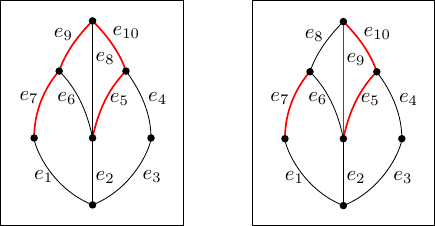}}
\caption{Instability of linked barcode.
Order of edges in the filtrations is indicated by subscripts
and positive edges are colored red. Left filtration 
has link bars $e_4-e_5, e_6-e_7, e_8-e_9, e_3-e_{10}$, and right filtration 
has link bars
$e_4-e_5,e_6-e_7,e_8-e_9$. With a transposition of the two edges $e_8$ and $e_9$ from the left to right, 
the bar $e_3-e_{10}$ of seven unit length vanishes.}
\label{fig:instability}
\end{figure}

The cause of this instability is rooted in the fact that
the split of positive and negative simplices $C_p(\cdot)=X_p(\cdot)\oplus Y_p(\cdot)$ changes with the filtrations. We fix it by considering this split in a ``reference'' filtration as in~\eqref{eq:filtration} which can be seen as growing
into the final complex $K:=K_m$. 
We consider a ``stabilized'' version of the link barcode
such that small perturbations on the reference filtration growing into $\cplx$
only cause small perturbations on the link barcode.
Let the ``reference'' simplex-wise filtration be
\begin{equation*}
\filtr: \cplxr_0\inctosp{\ssxr_1}\cplxr_1\inctosp{\ssxr_2}\cdots\inctosp{\ssxr_{m}} \cplxr_m=\cplx,
\end{equation*}
which determines the positivity and negativity for simplices in $\cplx$.

Let $\subcplx$ be a subcomplex of $\cplx$ and
$\pset^\filtr_p(\subcplx)$ (resp.\ $\nset^\filtr_p(\subcplx)$)
denote the set of $p$-simplices
of $\subcplx$ that are positive (resp.\ negative) in the filtration $\filtr$.
Define the subspace $X^\filtr_p(\subcplx)\subseteq C_p(\subcplx)$
to consist of all chains in $C_p(\subcplx)$ supported over 
$\pset^\filtr_p(\subcplx)$ and 
$Y^\filtr_p(\subcplx)\subseteq C_p(\subcplx)$
to consist of all chains in $C_p(\subcplx)$ supported over 
$\nset^\filtr_p(\subcplx)$.
We then have a split map similar to~\eqref{eqn:split},
\[\chi^\filtr_\subcplx:Z_p(\subcplx)\to X^\filtr_p(\subcplx)\oplus Y^\filtr_p(\subcplx)\] 
and in turn a projection map similar to~\eqref{eqn:proj},
\[\phi^\filtr_\subcplx:Z_p(\subcplx)\to Y^\filtr_p(\subcplx).\]

Now consider another subcomplex $\subcplx'$ of $\cplx$
such that $\subcplx\subseteq\subcplx'$. We have the following 
commutative diagram:
\begin{equation*}
\begin{tikzcd}[column sep=huge]
Z_p(\subcplx) \arrow[hookrightarrow,r,"\iG"] \arrow[d,"\phi^\filtr_{\subcplx}"] &[-5pt]  Z_p(\subcplx')  \arrow[d,"\phi^\filtr_{\subcplx'}"] \\
Y^\filtr_p(\subcplx) \arrow[hookrightarrow,r,"\iG"] & Y^\filtr_p(\subcplx') 
\end{tikzcd}
\end{equation*}
Therefore, the inclusion $\subcplx\subseteq\subcplx'$ induces
the following linear map:
\begin{equation}\label{eqn:alt-q-linmap}
\iG_*:Y^\filtr_p(\subcplx)\big/\,\img(\phi^\filtr_{\subcplx})\to Y^\filtr_p(\subcplx')\big/\,\img(\phi^\filtr_{\subcplx'}).
\end{equation}

\begin{definition}
Let $\subc_\cplx$ denote the category of subcomplexes of $\cplx$
with morphisms being inclusions. Define a functor 
\[\lfunc^\filtr_p:\subc_\cplx\to\vecc\]
where
$\lfunc^\filtr_p(\subcplx):=Y^\filtr_p(\subcplx)\big/\,\img(\phi^\filtr_{\subcplx})$
for each $\subcplx\in \subc_\cplx$
and $\lfunc^\filtr_p(\subcplx\subseteq\subcplx')$ is the induced linear map 
as in~\eqref{eqn:alt-q-linmap}
for $\subcplx\subseteq\subcplx'$ in $\subc_\cplx$.
\end{definition}

It is not hard to verify that $\lfunc^\filtr_p$ is a well-defined functor,
e.g., $\lfunc^\filtr_p(\subcplx\subseteq\subcplx'')=\lfunc^\filtr_p(\subcplx'\subseteq\subcplx'')\circ\lfunc^\filtr_p(\subcplx\subseteq\subcplx')$
for $\subcplx\subseteq\subcplx'\subseteq\subcplx''$
because the maps are induced by inclusions.
Notice that without a reference filtration for determining the positivity and negativity of simplices, 
the space $Y_p(\cdot)$ as in \Cref{dfn:posneg}
cannot be defined
and hence a functor like $\lfunc^\filtr_p$ for the stability result cannot be obtained.
\para{Stabilized link barcode.}
We now provide the alternative definition of link barcode 
which we shall later show to be stable.
When talking about stability, 
filtrations and persistence modules
are more naturally indexed over the category of real numbers $\Real$
instead of the category of integers $\mathbb{N}_0$.
For this, we consider
a \emph{monotone} function $f: \cplx \rightarrow \mathbb{R}$, where every pair of simplices $\sigma\subseteq \tau\in \cplx$ satisfies $f(\sigma)\leq f(\tau)$. 
We then have an $\Real$-indexed filtration $\filt_f:\Real\to\subc_\cplx$,
which is a functor such that $\filt_f(s)$ is the sublevelset
$\{\sigma\in\cplx\mid f(\sigma)\leq s\}$
and $\filt_f(s\leq t)$ is the inclusion.
Applying the functor $\lfunc^\filtr_p$ on $\filt_f$, we get an
$\Real$-module $\M^\filtr_{f,p}:=\lfunc^\filtr_p\circ\filt_f$,
which is a functor $\M^\filtr_{f,p}:\Real\to\vecc$ such that
\[\M^\filtr_{f,p}(s)=Y^\filtr_p\big(\filt_f(s)\big)\big/\,\img\big(\phi^\filtr_{\filt_f(s)}\big)\]
and $\M^\filtr_{f,p}(s\leq t)$ is the induced map.
It is well known~\cite{chazal2016structure} that 
a module like $\M^\filtr_{f,p}$ can be decomposed into submodules
supported over real-value intervals
similar to that in \Cref{eqn:intvmod}. The multi-set of real intervals
constitutes our alternative link barcode 
denoted $\lbarc^\filtr_p(f)$ which
could be retrieved as follows: 
Suppose that we have a simplex-wise filtration
\begin{equation}\label{eqn:tilde-filt-f}
\tilde{\filt}_f: \cplx_0\inctosp{\sigma_1} \cplx_1\inctosp{\sigma_2}\cdots\inctosp{\sigma_m} \cplx_m=\cplx
\end{equation}
where $i < j$ implies $f(\sigma_i)\leq f(\sigma_j)$.
Applying the functor $\lfunc^\filtr_p$,
we have an $\mathbb{N}_0$-module
\[\lfunc^\filtr_p\circ\tilde{\filt}_f: Y^\filtr_p(\cplx_0)\big/\,\img(\phi^\filtr_{\cplx_0})
\to Y^\filtr_p(\cplx_1)\big/\,\img(\phi^\filtr_{\cplx_1})\to\cdots\to
Y^\filtr_p(\cplx_m)\big/\,\img(\phi^\filtr_{\cplx_m})\]
which is similar to that in \Cref{dfn:lmod}.
For any finite integer interval
$[i,j)$ in the barcode of $\lfunc^\filtr_p\circ\tilde{\filt}_f$, 
one gets a real-value interval $[f(\sigma_i),f(\sigma_j))$
in $\lbarc^\filtr_p(f)$.
Moreover,
for any infinite interval
$[i,\infty)$ in the barcode of $\lfunc^\filtr_p\circ\tilde{\filt}_f$, 
one gets a real-value interval $[f(\sigma_i),\infty)$
in $\lbarc^\filtr_p(f)$.

\begin{theorem}
For any two monotone functions $f,g: \cplx \rightarrow \mathbb{R}$
and an integer $p>0$,
$d_B(\lbarc^\filtr_p(f),\lbarc^\filtr_p(g))\leq \|f-g\|_\infty$,
where $d_B$ is the bottleneck distance~\cite{CEH07} between barcodes
and $\|f-g\|_\infty=\max_{\sigma\in K} |f(\sigma)-g(\sigma)|$.
\end{theorem}
\begin{proof}
The proof follows the standard proof for stability of classical
persistence barcodes. We only briefly describe the idea; see, e.g.,~\cite{CEH07,chazal2016structure}
for details. Let $\eG=\|f-g\|_\infty$. Then, we have an \emph{$\eG$-interleaving}
between $\filt_f$ and $\filt_g$,
which in turn induces an $\eG$-interleaving between $\M^\filtr_{f,p}$
and $\M^\filtr_{g,p}$ by applying the functor $\lfunc$.
This implies that the \emph{interleaving distance} between the two $\Real$-modules
is no more than $\eG$. By the Isometry Theorem~\cite{chazal2016structure},
$d_B(\lbarc^\filtr_p(f),\lbarc^\filtr_p(g))\leq \eG$.
\end{proof}

Let $\Pers_p(f)$ denote the classical $p$-th persistence barcode
of the monotone function $f$.
We also define the alternative {linked barcode} of $f$,
denoted $\lkbarc^\filtr(f)$, as follows:
\begin{equation}\label{eqn:lkbr}
\lkbarc^\filtr(f):=\bigcup_{p}\Pers_p(f)\,\sqcup\,\bigcup_{p>0}\lbarc^\filtr_p(f).
\end{equation}
We can view $\lkbarc^\filtr(f)$ by drawing intervals in each $\Pers_p(f)$
as ``bars'' over the real line and drawing pairs in each $\lbarc^\filtr_p(f)$
as ``links'' between endpoints of the bars.

Notice that link bars in $\lbarc^\filtr_p(f)$ still connect death ends of
bars in $\Pers_{p-1}(f)$ with birth ends of bars in $\Pers_{p}(f)$.
However, some of the starting and ending points of the link bars defined without the reference filtration
(i.e., as in \Cref{dfn:lmod}) may disappear in the link bars in $\lbarc^\filtr_p(f)$.
For example, consider the filtration $\tilde{\filt}_f$ as in \eqref{eqn:tilde-filt-f}.
Suppose that $\sigma_i$,
which is added to $\cplx_{i-1}$ to obtain $\cplx_i$, 
is a negative $p$-simplex in $\filtr$.
Then,
the dimension increases from the vector space $Y^\filtr_p(\cplx_{i-1})$
to $Y^\filtr_p(\cplx_{i})$.
If $\sigma_i$ is also negative in $\tilde{\filt}_f$,
then $Z_p(\cplx_{i-1})=Z_p(\cplx_{i})$, 
meaning that $\img(\phi^\filtr_{\cplx_{i-1}})=\img(\phi^\filtr_{\cplx_{i}})$.
So, the dimension increases from
$\lfunc^\filtr_p(\cplx_{i-1})$ to $\lfunc^\filtr_p(\cplx_{i})$,
and we have a new starting point of a $p$-link bar in $\lbarc^\filtr_p(f)$
that is also the ending point of a $(p-1)$-bar in $\Pers_{p-1}(f)$.
However, if $\sigma_i$ is positive in $\tilde{\filt}_f$,
the dimension increases from
$\img(\phi^\filtr_{\cplx_{i-1}})$ to $\img(\phi^\filtr_{\cplx_{i}})$,
meaning that 
the dimension stays the same from
$\lfunc^\filtr_p(\cplx_{i-1})$ to $\lfunc^\filtr_p(\cplx_{i})$.
In this case, no starting point of a link bar is produced.


We then have:
\begin{theorem}
For any two monotone functions $f,g: \cplx \rightarrow \mathbb{R}$,
$d_B(\lkbarc^\filtr(f),\lkbarc^\filtr(g))\leq \|f-g\|_\infty$.
\label{thm:stable}
\end{theorem}
\begin{proof}
This follows from
$d_B(\lbarc^\filtr_p(f),\lbarc^\filtr_p(g))\leq \|f-g\|_\infty$
and $d_B(\Pers_p(f),\Pers_p(g))\leq \|f-g\|_\infty$
for any suitable $p$.
\end{proof}

\begin{remark}
   One may wonder how our remedy for stability works for the example in Figure~\ref{fig:instability}. Suppose that we take the filtration in the
    left as a reference filtration. Then, the bar 
    $e_8-e_9$ of one unit length disappears in the
    right filtration for the transposition of $e_8$ and $e_9$ and the bar $e_3-e_{10}$ of seven unit length
    remains intact.
\end{remark}

\para{Computation.}
 The algorithm for computing the stabilized link barcode 
remains the same with the only change that the matrix $V'$ as in \Cref{sec:alg}
is constructed differently. First, the graded boundary matrix $D$ for the given filtration is reduced as in Algorithm {\sc LinkedBar} with the first for loop.
Next, we compute the matrix $V'$ by retaining only the columns of positive simplices
in the given filtration and the rows of negative simplices in the reference filtration. This is easily implemented by keeping the second for loop of 
{\sc LinkedBar} the same while only changing the third for loop. The fourth for loop remains the same.

Let $R'$ be the reduced matrix of $V'$ as computed by the algorithm {\sc LinkedBar}.
Proposition~\ref{prop:link-reduce} holds for the linked barcode
obtained by this choice which at the same time remains stable
following Theorem~\ref{thm:stable}.

\cancel{
\section{Stability}
Now we argue that $p$-link barcodes are useful because they provide stable
persistence diagrams.
A piecewise constant function $f: K \rightarrow \mathbb{R}$ is called \emph{monotone}
if for every pair of simplices $\sigma, \tau\in K$ where $\sigma$ is a face of $\tau$, one has
$f(\sigma)\leq f(\tau)$. For such a monotone function $f$,
let $\sigma_1,\ldots,\sigma_m$ be an ordered sequence where
$i < j$ implies $f(\sigma_i)\leq f(\sigma_j)$. Then, we have a simplex-wise filtration
$\Fcal_f: K_0\inctosp{\sigma_1} K_1\inctosp{\sigma_2}\cdots\inctosp{\sigma_m} K_m=K$.
Consider two such filtrations $\Fcal_f$ and $\Fcal_g$ induced by functions $f, g: K\rightarrow \mathbb{R}$. The distance between 
$f$ and $g$ in the $\ell_\infty$-norm is given by
$\|f-g\|_\infty=\max_{\sigma\in K} |f(\sigma)-g(\sigma)|$. Following the standard
definition of persistence diagram~\cite{CEH07}, for $p>0$, we define the $p$-link persistence diagram $\Dcal_p(\Fcal_f)$
of the $p$-link module $\Lsf_p$ for the filtration $\Fcal_f$. It consists of a multiset of points
in the extended plane $\bar{\mathbb{R}}^2$ where $\bar{\mathbb{R}}=\mathbb{R}\cup \{\pm\infty\}$
along with points on the diagonal with infinite multiplicity.
For a finite bar
$[i,j)$ in the barcode of $\Lsf_p$, one gets a point $(a,b)$ where $f(\sigma_i)=a$ and $f(\sigma_j)=b$.
For an infinite bar $[i,\infty)$, one gets a point $(a,\infty)$ where $f(i)=a$.
Then, the proof of~\cite[Interpolation Lemma]{CEH07} can be 
adapted almost verbatim to claim the following stability in terms of the bottleneck distance $d_B(-,-)$
between the $p$-link persistence diagrams.
\begin{proposition}
    $d_B(\Dcal_p(\Fcal_f),\Dcal_p(\Fcal_g))\leq \|f-g\|_\infty$ for any $p> 0$.
\end{proposition}

More interestingly, we can define a \emph{linked} persistence diagram where we consider the
linked barcode of $\Fcal$ to define the diagram. Consider a linked bar $x[p,q]$ as defined
in section~\ref{sec:linkedbar}. In analogy to the usual persistence diagram~\cite{CEH07},
we map each such linked bar $x[p,q]$ to a point in the extended space $\bar{\mathbb{R}}^{2k+2}$ 
where $k$ is the dimension of the complex $K$ filtered by $\Fcal_f$. 
To do this, we add $p$ pairs of repeated
coordinates $(x_1^p,x_1^p)$ as a prefix and another $k-q$ pairs of repeated coordinates $(x_2^q,x_2^q)$
as suffix, that is, the linked bar $x[p,q]$ is assigned
to a point $(x_1^p,\ldots,x_1^p,x_1^p,x_2^p,\ldots,x_1^q,x_2^q,x_2^q,\ldots,x_2^q)$
in $\bar{\mathbb{R}}^{2k+2}$. Now we define the \emph{linked persistence diagram}
$\Dcal(\Fcal_f)$ as the multiset of points in $\bar{\mathbb R}^{2k+2}$, each corresponding to
a linked bar in $\Fcal_f$. We also add the diagonal line in $\bar{\mathbb R}^{2k+2}$ with
infinite multiplicity into the diagram. One can now define the bottleneck distance
$d_B(-,\cdots,-)$ between two linked persistence diagrams $\Dcal(\Fcal_f)$ and $\Dcal(\Fcal_g)$
by extending the standard definition for two dimensions. Again, the proof of~\cite[Interpolation Lemma]{CEH07} can be adapted verbatim to claim the stability result that
$d_B(\Dcal(\Fcal_f),\Dcal(\Fcal_g))\leq \|f-g\|_\infty$.

We can also consider persistence diagrams for higher order linked bars. For $k>1$, $p>0$, we define
the order-$k$ $p$-link persistence diagram $\Dcal^k_p(\Fcal_f)$
as the multiset of points in $\bar{\mathbb R}^{2}$ each corresponding to a $p$-link in the barcode
for $\Lsf_p^k$ along with the points on the diagonal. Again we have that
$d_B(\Dcal^k_p(\Fcal_f),\Dcal^k_p(\Fcal_g))\leq \|f-g\|_\infty$. Denoting the bottleneck distance
between linked
persistence diagrams as $d_B^1(f,g):=d_B(\Dcal(\Fcal_f),\Dcal(\Fcal_g))$, we define
$d_B^k(f,g):=\max(d_B^1(f,g), \max_{p>0,1<j\leq k}\{d_B(\Dcal^j_p(\Fcal_f),\Dcal^j_p(\Fcal_g))\}$.
The inequality $d_B^k(f,g)\leq \|f-g\|_\infty$ follows from the stability of each individual
component used in defining $d_B^k(f,g)$.
}

\section{Higher order $p$-link bars}
We have seen that 
the $p$-link bars connect a right endpoint of a bar of the
homology module $\Hm_{p-1}$ with the left endpoint of a bar of the homology module $\Hm_p$. We can extend the idea of $p$-link bars to \emph{higher order} $p$-link bars that pair some of the bars of $\Hm_p$ by connecting their \emph{left end points}. These points are the grades of generators in $\Z_p$, where new $p$-cycles are created.
In Figure~\ref{fig:linkedbar}, a bar (not drawn) connecting $e$ and $f$ is in second order.

We define the higher order link bars recursively by defining submodules of $\Z_p$ and
considering presentations of higher order linked modules using them. First, consider the
set of $p$-simplices $\pset_p(K_i)$ as defined in Definition~\ref{dfn:posneg}. Let $K_i^0=K_i$. Then, recursively define
$K_i^k$ as the smallest subcomplex of $K_i$ containing the simplices in $\pset_p(K_i^{k-1})$; e.g. along with their vertices, edge $d$, edges $\{d,e\}$, and edges $\{d,e,f\}$  in Figure~\ref{fig:linkedbar} constitute
$K^1_8$, $K^1_9$, and $K^1_{10}$ respectively.
The $k$th order subfiltration $\Fcal^k \subseteq \Fcal$ preserving the grades is:
\begin{equation*}
\Fcal^k: K^k_0\inctosp{\sigma_{k_1}}K^k_1\inctosp{\sigma_{k_2}}\ldots\inctosp{\sigma_{k_m}} K^k_{m}
\mbox{ where } \sigma_{k_{i}}  \mbox{ is either empty or equals } \sigma_{i} \mbox{ in Eq.~\eqref{eq:filtration}}.
\end{equation*}
All graded modules defined before can now be defined with this $k$th order filtration.
Let $\X_p^k$, $\Z_p^k$, and $\Y_p^k$ denote the respective $k$th order graded modules.
Consider the sequence $\Z_p^k\rightarrowsp{\iota}\X_p^k\oplus \Y_p^k \rightarrowsp{\pi} \Y_p^k$ which
gives the presentation of the $k$th order $p$-link module $\Lsf_p^k$:
\begin{equation*}
\Z_p^k\rightarrowsp{\pi\circ\iota}\Y_p^k\rightarrowsp{}\coker(\pi\circ\iota)=\Lsf_p^k
\end{equation*}

\subsection{Algorithm for computing higher order bars}
\label{sec:higherorderalgo}
To compute links of every order, we embed the algorithm {\sc LinkedBar} in a recursive routine {\sc OrdBar}. Suppose that we have already computed the
presentation matrix $D:=[\pi\circ\iota]$ for the presentation of the module $\Lsf_p^k$.
Then, we reduce $D$ to obtain the matrices $R$ and $V$ where $R=DV$. We read the barcode
of $\Lsf_p^k$ from the reduced matrix $R$. Then we choose a submatrix
of $V$ which selects (i) the columns that coincide (in terms of grades) with the columns in $R$
that are zeroed out (these columns represent a basis for $\Z^{k+1}_p$) and (ii) rows
whose grades coincide with those of the columns in $R$ that are not zeroed out (these
rows represent a basis for $\Y^{k+1}_p$. The algorithm {\sc OrdBar} calls itself on the
chosen submatrix of $V$ that presents the link module $\Lsf_p^{k+1}$ of the next order.
In the algorithms below, we combine the computations for all dimensions $0\leq p\leq k$ and
assume that $D$ is a block diagonal matrix, each
block (possibly empty) 
representing the presentation matrix $D_p$ for the module $\Lsf_p^k$.

The first time {\sc OrdBar} is called from the main routine {\sc OrdLinkedBar},
it is passed the graded boundary matrix $D$ of the input filtration. The algorithm
{\sc OrdBar} computes the barcode of the homology module $\Hm_p$ and then prepares
the matrix $V$ to be passed on to the next recursive call that begins computing
the linked barcodes. Each recursive pass of the algorithm takes at most $O(n^3)$ time
if $n$ is the size of the input filtration $\Fcal$.

Since every time we issue a
recursive call, the matrix being passed decreases by at least one column and one row,
there could be $O(n)$ such calls giving a total time complexity of $O(n^4)$. 

\begin{algorithm}
    \caption{{\sc{OrdBar}($D$)}}
    \KwData{A graded matrix $D$ with $n$ columns} 
    \KwResult{Links of an appropriate order}
    \If{$D$ is not empty column or empty row  matrix}{
        Initialize $V$ as an $n\times n$ identity matrix\; 
        \For{$j:=1$ to $n$}{ 
            \While{$\exists i<j \mbox{ so that } \low_D(j)=\low_D(i)$}{
                 $\col_D(j):=\col_D(j)+\col_D(i)$;
                $\col_V(j):=\col_V(j)+\col_V(i)$\;
            }
        }
    }
    Output barcode from reduced $D$ according to Proposition~\ref{prop:bars}\;
    \If{$D$ is not empty column or empty row matrix}{
    \For{every column $j$ of $V$ with $\col_D(j)$ non-zero}{
    $V:= V\setminus \col_V(j)$\;
    }
    \For{every row $i$ of $V$ with $\col_D(i)$ is zero}{
    $V:= V\setminus \row_V(i)$\;
    }
    Call {\sc OrdBar}($V$)
    }
\end{algorithm}

\begin{algorithm}
    \caption{{\sc{OrdLinkedBar}($\Fcal$)}}
    \KwData{A graded boundary matrix $D$ with $n$ columns built from $\Fcal$} 
    \KwResult{Links of $\Fcal$ of all orders}
    Initialize $V$ as an $n\times n$ identity matrix\; 
    \For{$j:=1$ to $n$}{ 
        \While{$\exists i<j \mbox{ so that } \low_D(j)=\low_D(i)$}{
             $\col_D(j):=\col_D(j)+\col_D(i)$;
            $\col_V(j):=\col_V(j)+\col_V(i)$\;
        }
    }
    Call {\sc OrdBar}($V$)
\end{algorithm}
The algorithm {\sc GraphLink} for graph filtration can also be adapted to generate
higher order linked barcodes. For this, we run {\sc GraphLink} on the input
filtration of the graph $G$, determine the edge set $X$ that creates $1$-cycles,
consider the subgraph $G'\subseteq G$ induced by $X$ and call {\sc GraphLink} on
the subfiltration of $G'$. We repeat the process until $G'$ becomes empty.
Since every time $G'$ shrinks by at least one edge, we repeat {\sc GraphLink} at most
$O(n)$ times giving a time complexity of $O(n^2\log n)$.

\begin{theorem}
    Given a filtration of a simplicial complex of
    size $n$, links of all orders can be computed in $O(n^4)$ time. If the
    complex is a graph, the complexity improves to $O(n^2\log n)$.
\end{theorem}
\section{Experiments for graph isomorphism problem}
We choose the graph isomorphism problem to illustrate the discriminating power of the linked diagrams in contrast
to the standard persistence diagrams. 
The problem is neither known to be in P  nor known to be NP-complete.
Babai \cite{10.1145/2897518.2897542} proposed
a pseudo-polynomial time algorithm, and also a polynomial algorithm that works in almost all cases excluding regular graphs \cite{doi:10.1137/0209047}. One of the most studied polynomial heuristics is the Weisfeiler-Leman algorithm aka.\  1-WL \cite{leman1968reduction}.
It is known that the standard persistent homology  can distinguish graphs as much as the 1-WL algorithm \cite{DBLP:conf/iclr/HornBMMRB22} does. However, 
regular graphs pose most difficulty for 1-WL and persistent homology
based approaches.
We show that our linked diagrams discriminate regular
graphs that are not distinguished by ordinary persistence diagrams.
In fact, it turns out that if we use the linked barcode
for higher orders, we get a better discrimination.

For a graph $G$, define vertex function
$f_v(G):V(G)\to \mathbb{R}$ that assigns to each vertex $v'\in V(G)$ the shortest path distance from $v$ in terms of the number of edges. Considering the lower star filtration 
of $G$ with $f_v$, we obtain the $k$th order $p$-link barcode $\lbarc_{p,k}(f_v)$
which together with the standard persistence barcode provide the $k$th order
$p$-linked barcode $\lkbarc_{p,k}(f_v)=\Pers_p(f_v)\sqcup\lbarc_{p,k}(f_v)$.
Taking a vertex $w$ on another
graph $H$, we 
obtain the bottleneck distance $d_B^k(f_v,f_w)=\max_p\{d_B(\lkbarc_{p,k}(f_v(G)),\lkbarc_{p,k}(f_w(H)))\}$ which
we take as a distance between $G$ and $H$:
\begin{equation*}
\label{eq:graph-distance}
h^k(G,H) := \max\{\max_{v\in G} \min_{w\in H}d_B^k(f_v,f_w), \max_{w\in H} \min_{v\in G}d_B^k( f_v,f_w)\}
\end{equation*}

\begin{table}[thb!]
    \centering
    \begin{tabular}{|c|c|c|c|c|c|c|c|}
    \hline
        Graph family &Total number of pairs & $h^0$ & $h^1$ & $h^2$ & $h^3$ & $h^4$ & $h^5$ \\ \hline
        (9,4,3)& 120     & 109   & 118   & -     & -     & -     & -     \\ \hline
        (10,3,3)&171    & 168   & 171   & -     & -     & -     & -     \\ \hline
        (12,4,4)& 66     & 44    & 55    & 66    & -     & -     & -     \\ \hline
        (13,4,4)& 465    & 386   & 415   & 464   & -     & -     & -     \\ \hline
        (16,6,2,2)& 1    & 1     & -     & -     & -     & -     & -     \\ \hline
        (25,12,5,6)& 105 & 0     & 14    & 59    & 94    & 105   & -     \\ \hline
        (26,10,3,4)& 45  & 0     & 9     & 42    & 44    & 45    & -     \\ \hline
    \end{tabular}
    \caption{Graph families in first 4 rows denoted $(n,g,d), \#$ have $n$ vertices, girth size at least $g$, and are $d$ regular. Graph families in the last 3 rows denoted $(n,d,\lambda,\mu), \#$ are strongly $d$
    regular where two adjacent vertices share $\lambda$ neighbors and each two non-adjacent vertices share $\mu$ neighbors. Each $h^k$ column shows the number of discriminated pairs of graphs using the barcode up to the $k$th order. When this numbers remains the same it is denoted with a $-$ sign. } 
    \label{tb:conglomerate-tables}
\end{table}



\cancel{
\subsection{Results}
One of the most used and well known algorithms for graph isomorphism testing is the Weisfeiler-Leman graph isomorphism test (WL-1), or one of its variants, these have the advantage to run in  polynomial time. However, these algorithms don't work well when applied to regular graphs, specially strong regular graphs. Other non polynomial time algorithms like nauty \cite{MCKAY201494}, have been proposed and work really well on practice. In this subsection we test how powerful   Proposition \ref{distance:graphs} is to discriminate regular and strong regular graphs. Specifically we are interested in how much is gained by using the linked barcode instead of the usual persistence barcode in the definition of (\ref{eq:graph-distance}).

We do our test in the set of all regular graphs with 10 vertices and degree 3. For each pair of graphs we computed two distances, using the usual barcode and the linked barcode, with their corresponding bottleneck distances as defined in   equation (\ref{eq:graph-distance}), the results are shown in table.
\definecolor{bluey}{rgb}{0.25, 0.67, 96}
\definecolor{green}{rgb}{0.28, 0.61, 0.14}
\definecolor{red}{rgb}{0.81960784, 0.09803922, 0.01568627}
\begin{table}[!htb]
    
    \begin{minipage}{.5\linewidth}

      \centering
    \begin{tabular}{|l|l|l|l|l|l|l|l|l|l|}
    \hline
 
        \diagbox[height=1cm,innerwidth=1cm]{$d_B^0}{$d_B^1$} &1 &2&3& 4 &  5 & 6 & 7 & 8 &9  \\ \hline
        1 &\diagbox[height=0.64cm,innerwidth=.31cm]{}{}  & 1 & 1  & \cellcolor{green}1 &  1 &  1 &  1 & 1 &1\\ \hline
        2 & 1 & {\diagbox[height=0.64cm,innerwidth=.31cm]{}{}}&  \cellcolor{green}1 & 1 &  1& 1 & 1 &  \cellcolor{green}1& \cellcolor{green}1 \\ \hline
        3 & 1 &  \cellcolor{red}0 & {\diagbox[height=0.64cm,innerwidth=.31cm]{}{}}& 1 & 1 & 1&  1&\cellcolor{green}1 & \cellcolor{green}1\\ \hline
        4   &  \cellcolor{red}0 & 1 & 1& {\diagbox[height=0.64cm,innerwidth=.31cm]{}{}} & 1& 1 & 1 & 1 & 1  \\ \hline
        5 & 1 & 1 & 1 & 1 &{\diagbox[height=0.64cm,innerwidth=.31cm]{}{}}  &\cellcolor{green}1 & 1 & 1 & 1 \\ \hline
        6 & 1 & 1 & 1 & 1 &  \cellcolor{red}0&{\diagbox[height=0.64cm,innerwidth=.31cm]{}{}}& 1 & 1 & 1\\ \hline
        7 & 1 & 1 & 1 & 1  & 1 & 1 & {\diagbox[height=0.64cm,innerwidth=.31cm]{}{}} & 1 & 1 \\ \hline
        8 & 1& \cellcolor{red}0 &\cellcolor{red}0 &1 & 1 &1 & 1 & {\diagbox[height=0.64cm,innerwidth=.31cm]{}{}} & \cellcolor{bluey}0\\ \hline
        9 &1 & \cellcolor{red}0 & \cellcolor{red}0 & 1 & 1 & 1 & 1 & \cellcolor{red}0 & {\diagbox[height=0.64cm,innerwidth=.31cm]{}{}}\\ \hline

    \end{tabular}
    \end{minipage}%
\end{table}
}
\cancel{
\begin{table}
    \begin{tabular}{|c|c|c|c|c|}
    \hline
        Graph family &No graphs & $d_B^0$ & $d_B^1$&  $d_B^2$\\ \hline
        (9,4,3) & 16 &109 &116 &120 \\ \hline
        (10,3,3)& 19 &160& 171&171 \\ \hline
        (12,4,4)& 12 &44 &66 & 66\\ \hline
        (13,4,4)& 31 & 386 & 449 &463 \\ \hline
        (16,6,2,2)& 2 & 1&1 &1 \\ \hline
        (25,12,5,6)& 15 &0 & 14& 27** \\ \hline
    \end{tabular}

\end{table}
}
\Cref{tb:conglomerate-tables} shows the number of pairs of non-isomorphic graphs detected by distance up to $h^k$ in each row; graph families are obtained from \cite{Coolsaet2023}. 
Every family except (9,4,3) and (13,4,4) is completely distinguished for non-isomorphic pairs after a certain order is reached beyond which
the number of pairs of discriminated graphs remain 
the same (denoted as $-$). Notice that
$h^0$ is the distance obtained by ordinary persistence diagrams.

\subsection{Applications to temporal network link prediction}

The problem of temporal link prediction is a foundational challenge in dynamic network analysis, centering on the ability to anticipate future structural changes in a network based on its historical topology. We frame temporal link prediction as a binary classification task to determine whether a new link will form between two nodes. Given a temporal network of timestamped edges $\mathcal{S} = \{(u_i, v_i, t_i)\}$, we sort the interactions chronologically and split them at a cutoff time $t^\star$ so that approximately 90\% of the earliest data forms the training set ($E_{\text{train}}$) and the remaining 10\% forms the test set ($E_{\text{test}}$). For each dataset, we consider only the first occurrence of any undirected pair to ensure we predict new link formation rather than recurrence, and we remove self-loops.

To train and evaluate our models, we construct balanced, 50/50 target sets of positive and negative pairs. For every true positive edge $(u,v,t)$ in a split, we generate a corresponding negative pair by sampling a non-edge $(u', v', t)$ uniformly at random from the network's vertices. We reject any negative candidate that ever appears as an actual edge in the full stream up to time $t$, and we assign it the exact same timestamp $t$ as its paired positive. This ensures that the negative examples inherit an identical temporal distribution to the true edges, forcing the model to learn structural patterns rather than simple timing cues.

For each candidate pair $(u,v,t)$, we extract a localized, time-causal subgraph $G_{uv}^{(t)}$ that captures only the historical network structure knowable strictly \textit{before} the reference time $t$. This is achieved using  a $2$-hop ego-net. We then construct a lower-star filtration on these subgraphs by normalizing the node arrival times into a range of $[0, 1]$. We compute three distinct topological barcodes per pair: barcodes for standard $H_0$, $H_1$, and the link barcode for $\sf L_1$.

Each barcode is vectorized into a fixed-length feature block using persistence images \cite{persistenceimage}. To isolate the explicit predictive value of the link barcode, we train a random forest classifier \cite{randomforest} over 20 independent seeds and compare two setups. The \emph{base} setup relies solely on the $H_0$ and $H_1$ vectors, which we optionally augment with two classical, time-restricted link prediction heuristics \cite{libennowell2007}: Preferential Attachment (PA) \cite{PA} and the Adamic–Adar index (AA) \cite{AA}. The \emph{link} setup adds the link-persistence feature block ($\sf L_1$) to these baselines.

Crucially, the classifier's internal random seed is kept identical between paired runs; leaving the presence or absence of the $L_1$ features as the only moving part. Adding the link feature increases AUC in most experiments (Table~\ref{tab:link_prediction}), with one exception (HighSchool with heuristics), where strong classical heuristics already saturate performance. We report the mean Area Under the ROC Curve (AUC) and use a paired $t$-test \cite{student1908probable} to determine if the topological link features provide a statistically significant performance boost over traditional network heuristics see Table \ref{tab:link_prediction}.
\begin{table}[h!]
\centering
\caption{Mean AUC over 20 seeds with and without the $L_1$ link features. 
$\Delta$ is reported in percentage points (\%); p-values are from a paired 
$t$-test between the two setups.}
\label{tab:link_prediction}
\begin{tabular}{llcccc}
\toprule
\textbf{Dataset} & \textbf{Heuristics} & \textbf{Base AUC} & \textbf{Link AUC} & \textbf{$\Delta$ (\%)} & \textbf{p-value} \\
\midrule
email-Eu      & none  & 0.549 & 0.577 & $+2.8\%$   & $<$0.001 \\
email-Eu      & AA+PA & 0.683 & 0.690 & $+0.7\%$   & 0.21     \\
HighSchool    & none  & 0.565 & 0.580 & $+1.5\%$   & $<$0.001 \\
HighSchool    & AA+PA & 0.805 & 0.798 & $-0.7\%$   & $<$0.001 \\
UCI msg       & none  & 0.539 & 0.552 & $+1.3\%$   & 0.14     \\
UCI msg       & AA+PA & 0.531 & 0.546 & $+1.5\%$   & 0.001    \\
Enron core200 & none  & 0.555 & 0.562 & $+0.7\%$   & 0.002    \\
Enron core200 & AA+PA & 0.609 & 0.616 & $+0.7\%$   & 0.001    \\
fb-forum      & none  & 0.534 & 0.539 & $+0.5\%$   & 0.009    \\
fb-forum      & AA+PA & 0.536 & 0.542 & $+0.6\%$   & 0.001    \\
CollegeMsg    & none  & 0.553 & 0.558 & $+0.5\%$   & 0.018    \\
CollegeMsg    & AA+PA & 0.515 & 0.525 & $+1.0\%$   & 0.018    \\
\bottomrule
\end{tabular}
\end{table}
\newpage


\bibliography{ref}

\end{document}